\documentclass{article}

\usepackage{microtype}
\usepackage{graphicx}
\usepackage{subfigure}
\usepackage{booktabs} 

\usepackage{hyperref}

\usepackage[accepted]{icml2021}

\usepackage{algorithm} 
\usepackage{xcolor}
\usepackage{amsmath, array, makecell, amsfonts, amsthm}
\usepackage{mathrsfs}
\usepackage{cellspace}
\usepackage{enumitem}
\usepackage{bm}
\usepackage[T1]{fontenc}
\usepackage[algoruled,boxed,lined,noend,algo2e]{algorithm2e}
\usepackage{booktabs}

\DeclareMathOperator*{\argmax}{arg\,max}

\newtheorem{definition}{Definition}
\newtheorem{theorem}{Theorem}

\SetKwRepeat{Do}{do}{while}%

\setlist[description]{leftmargin=\parindent,labelindent=\parindent}

\newcommand{\NEW}[1]{#1}

\icmltitlerunning{Learning to Play Against Any Mixture of Opponents}

\begin{document}

\twocolumn[
\icmltitle{Learning to Play Against Any Mixture of Opponents}



\icmlsetsymbol{equal}{*}

\begin{icmlauthorlist}
\icmlauthor{Max Olan Smith}{umich}
\icmlauthor{Thomas Anthony}{deepmind}
\icmlauthor{Yongzhao Wang}{umich}
\icmlauthor{Michael P. Wellman}{umich}
\end{icmlauthorlist}

\icmlaffiliation{umich}{University of Michigan}
\icmlaffiliation{deepmind}{Deepmind}

\icmlcorrespondingauthor{Max Olan Smith}{mxsmith@umich.edu}

\icmlkeywords{Reinforcement Learning, Multiagent Learning, Transfer Learning}

\vskip 0.3in
]



\printAffiliationsAndNotice{} 

\begin{abstract}
Intuitively, experience playing against one mixture of opponents in a given domain should be relevant for a different mixture in the same domain. 
We propose a transfer learning method, \textit{Q-Mixing}, that starts by learning Q-values against each pure-strategy opponent. Then a Q-value for \emph{any} distribution of opponent strategies is approximated by appropriately averaging the separately learned Q-values.
From these components, we construct policies against all opponent mixtures without any further training. 
We empirically validate Q-Mixing in two environments: a simple grid-world soccer environment, and \NEW{a social dilemma game}. 
We find that Q-Mixing is able to successfully transfer knowledge across any mixture of opponents. 
We next consider the use of observations during play to update the believed distribution of opponents.
We introduce an opponent classifier---trained in parallel to Q-learning, reusing  data---and use the classifier results to refine the mixing of Q-values. 
We find that Q-Mixing augmented with the opponent policy classifier performs \NEW{better, with higher variance,} than training directly against a mixed-strategy opponent.
\end{abstract}

\section{Introduction}
\label{sec:introduction}

Reinforcement learning (RL) agents commonly interact in environments with other agents, whose behavior may be uncertain.
For any particular probabilistic belief over the behavior of another agent (henceforth, \emph{opponent}), we can learn to play with respect to that opponent distribution (henceforth, \emph{mixture}), for example by training in simulation against opponents sampled from the mixture.
If the mixture changes, ideally we would not have to train from scratch, but rather could \emph{transfer} what we have learned to construct a policy to play against the new mixture.

\begin{figure}[!ht]
    \centering
    \includegraphics[width=0.45\textwidth]{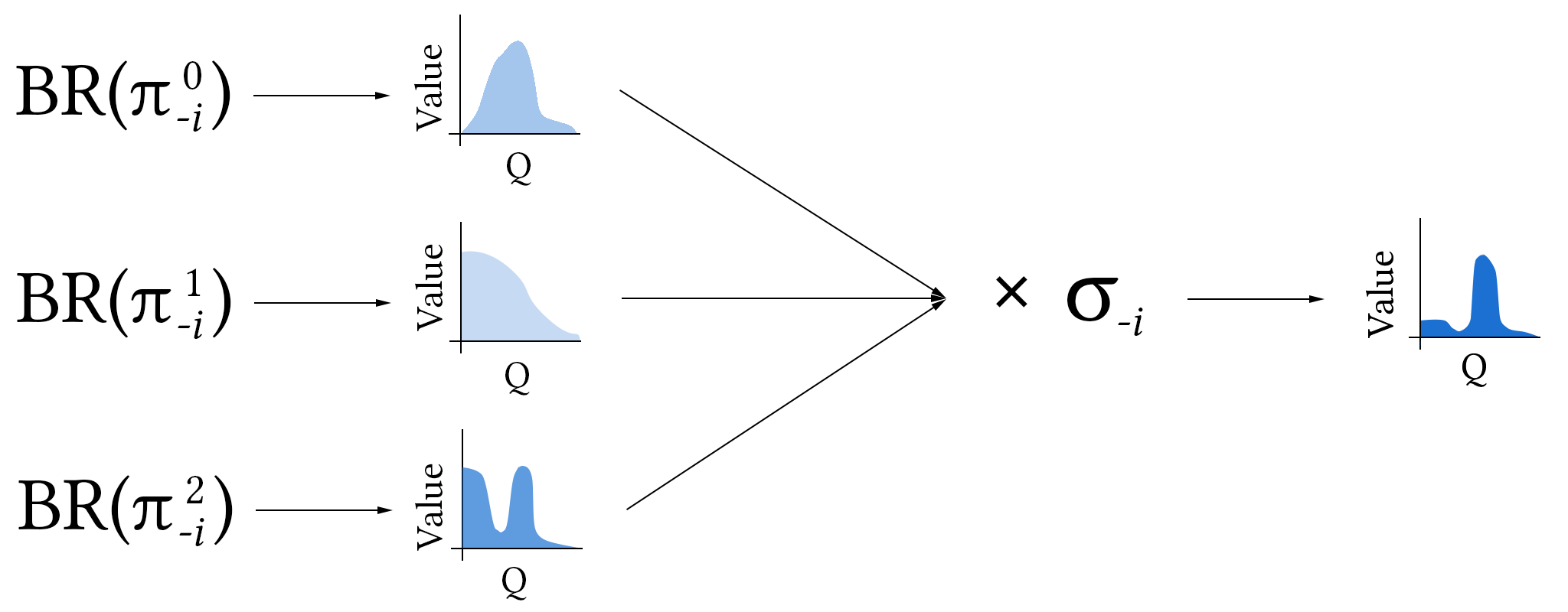}
    \caption{
        Q-Mixing concept. 
        BR to each of the pure strategies~$\pi_{-i}$ are learned separately. 
        The Q-value for a given opponent mixed strategy~$\sigma_{-i}$ is then derived by combining these components.
    }
    \label{fig:teaser}
\end{figure}

Traditional RL algorithms include no mechanism to explicitly prepare for variability in opponent mixtures. 
Instead, current solutions either learn a new behavior or update a previously learned behavior.
In lieu of that, researchers designed methods for learning a single behavior successfully across a set of strategies~\cite{wang03pareto}, or quickly adapting in response to new strategies~\cite{jaderberg19ctf}. 
In this work, we explicitly tackle the unique problem of responding to new opponent mixtures without requiring further simulations for learning.

We propose a new algorithm, \emph{Q-Mixing}, that effectively transfers learning across opponent mixtures.
The algorithm is designed within a population-based learning regime \cite{jaderberg19ctf, lanctot17psro}, where training is conducted against a known distribution of opponent policies.
Q-Mixing initially learns value-based best responses (BR), represented as Q-functions, with respect to each of the opponent's pure strategies.
It then transfers this knowledge against any given opponent mixture by weighting the Q-functions according to the probability that the opponent plays each of its pure strategies.
The idea is illustrated in Figure~\ref{fig:teaser}.
This calculation is an approximation of the Q-values against the mixed-strategy opponent, with error due to misrepresenting the future belief of the opponent by the current belief.
The result is an approximate BR policy against any mixture, constructed without additional training.

\NEW{%
This situation is common in game-theoretic approaches to multiagent RL such as Policy-Space Response Oracles (PSRO)~\cite{lanctot17psro} where a population of agents is trained and reused.
\citet{smith21} showed that combining PSRO and Q-Mixing results in large reductions of PSRO's cumulative training time.}

We experimentally validate our algorithm on: (1)~a simple grid-world soccer game, and (2)~a sequential social dilemma game. 
Our experiments show that Q-Mixing is effective in transferring learned responses across opponent mixtures. 

We also address two potential issues with combining Q-functions according to a given mixture. 
The first is that the agent receives information \emph{during play} that provides evidence about the opponent's strategy. 
We address this by introducing an opponent classifier to predict which opponent we are playing for a particular state and reweighting the Q-values to focus on the likely opponents.
The second issue is that the complexity of the policy produced by Q-Mixing grows linearly with the support of the opponent's mixture. 
This can make simulation and decision making computationally expensive. 
\NEW{We propose and show that policy distillation~\cite{rusu16distill} can compress a Q-Mixing policy while maintaining performance.}

\smallskip\noindent\textbf{Key Contributions:}%
(1)~Theoretically relate (in idealized setting) the Q-value for an opponent mixture to Q-values for mixture components.
(2)~A new transfer learning algorithm, \emph{Q-Mixing}, that uses this relationship to construct approximate policies against any given opponent mixture, without additional training.
(3)~Augmenting this algorithm with runtime opponent classification, to account for observations that can inform predictions of opponent policy during play.
(4)~Demonstrate that policy distillation can reduce Q-Mixing's complexity (in both memory and computation).

\section{Preliminaries}
\label{sec:preliminaries}
At any time $t\in\mathcal{T}$, an agent receives the environment's state $s^t\in\mathcal{S}$, or an observation $o^t\in\mathcal{O}$, a partial state.
(Even if an environment's state is fully observable, the inclusion of other agents with uncertain behavior makes the overall system partially observable.)
\NEW{From said observation, the agent takes an action} $a^t\in\mathcal{A}$ \NEW{receiving} a reward $r^t\in\mathbb{R}$.
The agent's \emph{policy} describes its behavior given each observation $\pi:\mathcal{O}\to\Delta(\mathcal{A})$\@. 
Actions are received by the environment, and a next observation is determined following the environment's transition dynamics $p:\mathcal{O}\times\mathcal{A}\to\mathcal{O}$\@. 

The agent's goal is to maximize its reward over time; called the \emph{return}: $G^t=\sum_{l}^{\infty}\gamma^{l}r^{t+l}$, where $\gamma$ is the discount factor weighting the importance of immediate rewards. 
Return is used to define the values of being in of a given observation: 
\begin{displaymath}
V(o^t) = \mathbb{E}_{\pi}\left[\sum_{l=0}^{\infty} \gamma^{l}r(o^{t+l}, a^{t+l}) \right],
\end{displaymath}
and taking an action given an observation: 
\begin{displaymath}
Q(o^t, a^t) = r(o^t, a^t) + \gamma \mathbb{E}_{o^{t+1}\in\mathcal{O}}\left[ V(o^{t+1}) \right].
\end{displaymath}

For multiagent settings, we index the agents and distinguish the agent-specific components with subscripts. 
Agent~$i$'s policy is $\pi_i:\mathcal{O}_i\to\Delta(\mathcal{A}_i)$, and the opponent's policy%
\footnote{Our methods are defined here for environments with two agents. 
Extension to greater numbers while maintaining computational tractability is a topic for future work.}
is the negated index, $\pi_{-i}:\mathcal{O}_{-i}\to\Delta(\mathcal{A}_{-i})$.
Boldface elements are joint across the agents (e.g., joint-action $\bm{a}$).

Agent~$i$ has a strategy set $\Pi_i$ comprising the possible policies it can employ.
Agent~$i$ may choose a single policy from $\Pi_i$ to play as a \emph{pure strategy}, or randomize with a \emph{mixed strategy} $\sigma_i\in\Delta(\Pi_i)$. 
Note that the pure strategies $\pi_i$ may themselves choose actions stochastically.
For a mixed strategy, we denote the probability the agent plays a particular policy~$\pi_{i}$ as $\sigma_{i}(\pi_{i})$. 
A \emph{best response} ($\text{BR}$) to an opponent's strategy $\sigma_{-i}$ is a policy with maximum return against $\sigma_{-i}$.

Agent~$i$'s prior belief about its opponent is represented by an opponent mixed-strategy, $\sigma_{-i}^0 \equiv \sigma_{-i}$. 
The opponent plays the mixture by sampling a policy according to $\sigma_{-i}$ at the start of the episode.
They are locked into the sampled policy's behavior for the entire duration of the episode.
The agent's updated belief at time~$t$ about the opponent policy faced is denoted $\sigma_{-i}^t$.

We introduce the term \textit{Strategy Response Value} (SRV) to refer to the observation-value against an opponent's strategy.
\begin{definition}[Strategic Response Value]
An agent's $\pi_i$ \emph{strategic response value} is its expected return given an observation, when playing $\pi_i$ against a specified opponent strategy:
\begin{displaymath}
V_{\pi_i}(o^t_i \mid \sigma_{-i}^t)=\mathbb{E}_{\sigma_{-i}^t}\left[ \sum_{\bm{a}}\bm{\pi}(a_i\mid o_i^t) \sum_{o_i',r_i} p(o_i', r_i\mid o_i^t, \bm{a}) \delta \right]
\end{displaymath}
where  $\delta\equiv r_i + \gamma V_{\pi_i}(o_i'\mid\sigma_{-i}^{t+1})$.
Let the \emph{optimal SRV} be 
\begin{displaymath}
V^*_i(o^t_i \mid \sigma_{-i}^t)= \max_{\pi_i}V_{\pi_i}(o^t_i \mid \sigma_{-i}^t).
\end{displaymath}
\label{defn:srv}
\end{definition}

From the SRV, we define the Strategic Response Q-Value (SRQV) for a particular opponent strategy.
\begin{definition}[Strategic Response Q-Value]
\label{defn:srqv}
An agent's $\pi_i$ \emph{strategic response Q-value} is its expected return for an action given an observation, when playing $\pi_i$ against a specified opponent strategy: 
\begin{displaymath}
Q_{\pi_i}(o_i^t, a_i^t \mid \sigma_{-i}^t) = \mathbb{E}_{\sigma_{-i}^t}\left[ r_i^t\right] + \gamma \mathbb{E}_{o^{t+1}_i}\left[ V_{\pi_i}(o^{t+1}_i \mid \sigma_{-i}^{t+1}) \right],
\end{displaymath}
where $r_i^t\equiv r_i(o_i^t,a_i^t,a_{-i}^t)$.
Let the \emph{optimal SRQV} be 
\begin{displaymath}
Q^*_i(o_i^t, a_i^t \mid \sigma_{-i}^t) = \max_{\pi_i}Q_{\pi_i}(o_i^t, a_i^t \mid \sigma_{-i}^t).
\end{displaymath}
\end{definition}

\section{Q-Mixing}
\label{sec:qmixing}

Our goal is to transfer the Q-learning effort across different opponent mixtures.
We consider the scenario where we first learn against each opponent's pure strategy. 
From this, we construct a Q-function for a given distribution of opponents from the policies trained against each opponent's pure strategy. 

\subsection{Single-State Setting}
Let us first consider a simplified setting with a single state.
This is essentially a problem of bandit learning, where our opponent's strategy will set the reward of each arm for an episode. 
Intuitively, our expected reward against a mixture of opponents is proportional to the payoff against each opponent weighted by their respective likelihood. 

As shown in Theorem~\ref{thm:qmix}, weighting the component SRQV by the opponent's distribution supports a BR to that mixture. 
We call this relationship \emph{Q-Mixing-Prior} and define it in Theorem~\ref{thm:qmix-single-state} (proof provided in Section~\ref{sec:q-mixing-proof}).
\begin{theorem}[Single-State Q-Mixing]
\label{thm:qmix-single-state}
    Let $Q^*_i(\cdot \mid \pi_{-i})$, $\pi_{-i}\in\Pi_{-i}$, denote the optimal strategic response Q-value against opponent policy $\pi_{-i}$. 
    Then for any opponent mixture $\sigma_{-i} \in\Delta(\Pi_{-i})$, the optimal strategic response Q-value is given by
\begin{displaymath}
Q^*_i(a_i\mid\sigma_{-i}) = \sum_{\pi_{-i}\in\Pi_{-i}} \sigma_{-i}(\pi_{-i}) Q^*_i(a_i\mid\pi_{-i}).
\end{displaymath}
\end{theorem}

\subsection{Leveraging Information from the Past}
Next, we consider the RL setting where both agents are able to influence an evolving observation distribution.
As a result of the joint effect of agents' actions on observations, the agents have an opportunity to gather information about their opponent during an episode.
Methods in this setting need to (1)~use information from the past to update its belief about the opponent, and (2)~grapple with uncertainty about the future.
\NEW{To bring Q-Mixing into this setting we need to quantify the agent's current belief about their opponent and their future uncertainty.}

\NEW{%
During a run with an opponent's pure strategy drawn from a distribution, the actual observations experienced generally depend on the identity of this pure strategy.
Let $\psi_i: \mathcal{O}_{i}^{0:t}\to \Delta(\Pi_{-i})$ represent the agent's current belief about the opponent's policy using the observations during play as evidence. 
From this prediction, we propose an approximate version of Q-Mixing that accounts for past information.
The approximation works by first predicting the relative likelihood of each opponent given the current observation.
Then it weights the Q-value-based BRs against each opponent by their relative likelihood.}

Figure~\ref{fig:uncertainty} provides an illustration of the benefits and limitations of this new prediction-based Q-Mixing.
At any given timestep $t$ during the episode, the information available to an agent about the opponents may be insufficient to perfectly identify their policy.
The yellow area above a timestep represents the uncertainty reduction from an updated prediction of the opponent $\sigma_{-i}^t$ compared to the baseline prediction of the prior $\sigma_{-i}^0$.
Crucially, this definition of Q-Mixing does not consider updating the opponent distribution from new information in the future (blue area in Figure~\ref{fig:uncertainty}).

\begin{figure}[ht]
    \centering
    \includegraphics[width=0.4\textwidth]{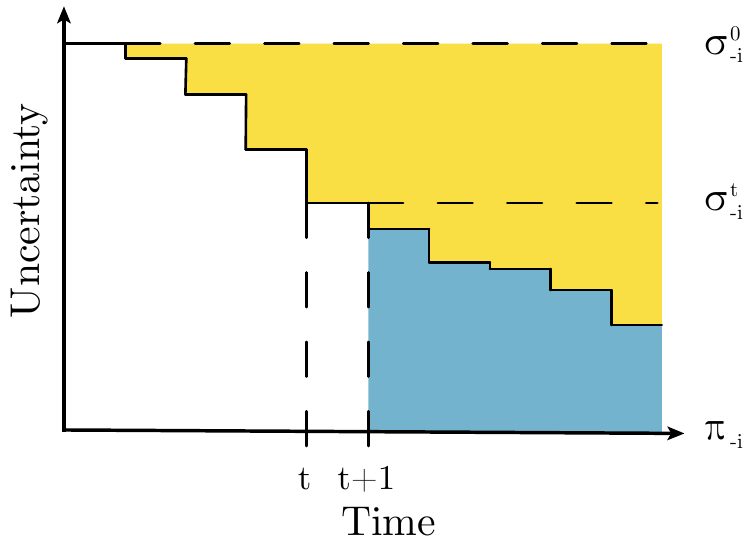}
    \caption{
        Opponent uncertainty over time. 
        The yellow area represents uncertainty reduction as a result of updating belief about the distribution of the opponent.
        The blue area represents approximation error incurred by Q-Mixing.
    }
    \label{fig:uncertainty}
\end{figure}

Let the previously defined $\psi$ be the \emph{opponent policy classifier} (OPC), which predicts the opponent's policy. 
In this work, we consider a simplified version of this function, that operates only on the agent's current observation $\psi_i:\mathcal{O}_i\to\Delta(\Pi_{-i})$.
We then augment Q-Mixing to weight the importance of each BR as follows:
\begin{displaymath}
Q_{\pi_i}(o_i, a_i \mid \sigma_{-i}) = \sum_{\pi_{-i}}\psi_i(\pi_{-i}\mid o_i, \sigma_{-i})Q_{\pi_i}(o_i, a_i \mid \pi_{-i}).
\end{displaymath}
We refer to this quantity as \emph{Q-Mixing}, or \emph{Q-Mixing-X}, where X describes $\psi$ (e.g., \emph{Q-Mixing-Prior} defines $\psi$ as the prior uncertainty of the opponent $\sigma_{-i}^0$). 
By continually updating the opponent distribution during play, the adjusted Q-Mixing result better responds to the actual opponent.

An ancillary benefit of the opponent classifier is that poorly estimated Q-values tend to have their impact minimized. 
For example, if an observation occurs only against the second pure strategy, then the Q-value against the first pure strategy would not be trained well, and thus could distort the policy from Q-Mixing.
These poorly trained cases correspond to unlikely opponents and get reduced weighting in the version of Q-Mixing augmented by the classifier. 

\subsection{Accounting for Future Uncertainty}
\NEW{To account for future uncertainty, Q-Mixing must be able to update its successor observation values given future observations.}
This can be done by expanding the Q-value \NEW{into its components}: expected reward under the current belief in the opponent's policy, and our expected next observation value.
\NEW{By updating the second term to recursively reference a new opponent belief we can account for changing beliefs in the future.}
The extended formulation, \NEW{Q-Mixing-Value-Iteration (QMVI)}, is given by:
\begin{align*}
    \delta &= r_i(o_i^t, a_i^t \mid \pi_{-i}) + \gamma \mathbb{E}_{o_i^{t+1}}\left[ V^*(o_i^{t+1}\mid\sigma_{-i})\right], \\
    Q^*_i&(o_i^t,a_i^t\mid\sigma_{-i}) = \sum_{\pi_{-i}\in\Pi_{-i}} \psi_i(\pi_{-i}\mid o_i^t, \sigma_{-i}) \cdot \delta . \\
\end{align*}

If we assume that we have access to both a dynamics model and the observation distribution dependent on the opponent, then we can directly solve for this quantity through Value Iteration (Algorithm~\ref{alg:vi}). 
These requirements are quite strong, essentially requiring perfect knowledge of the system with regards to all opponents.
The additional step of Value Iteration also carries a computational burden, as it requires iterating over the full observation and action spaces.
Though these costs may render \NEW{QMVI} infeasible in practice, we provide Algorithm~\ref{alg:vi} in Section~\ref{sec:q-mixing-proof} as a way to ensure correctness in Q-values.

\section{Experiments}
\label{sec:experiments}

\subsection{Grid-World Soccer}
We first evaluate Q-Mixing on a simple grid-world soccer environment \cite{littman94markov, greenwald03correlatedq}.
This environment has small state and action spaces, allowing for inexpensive simulation.
With this environment we pose the following questions:
(1)~Can QMVI obtain Q-values for mixed-strategy opponents?
(2)~Can Q-Mixing transfer Q-values across all of the opponent's mixed strategies?

The soccer environment is composed of a soccer field, two players, one ball, and two goals. 
The player's objective is to acquire the ball and score a goal while preventing the opponent from scoring on their own goal. 
The scorer receives $+1$ reward, and the opponent receives $-1$ reward. 
The state of the environment consists of the entire field including the locations of the players and ball.
This is represented as a $5\times4$ matrix with six possible values in each cell, referring to the occupancy of the cell.
The players may move in any of the four cardinal directions or stay in place.
Actions taken by the players are executed in a random order, and if the player possessing the ball moves last then the first player may take possession of the ball by colliding with them. 
A graphical example of the soccer environment can be seen in Figure~\ref{fig:soccer}.

In our experiments, we learn policies for Player~1 using Double DQN \cite{hasselt16ddqn}.
The state space is ravelled into a vector of length 120 that is then fed into a deep neural network.
The network has two fully-connected hidden layers of size 50 with ReLU activation functions and an output layer over the actions with size 5.
Player~2 plays a strategy over five policies, \NEW{each using the same shape neural networks} as Player~1, generated using the double oracle (DO) algorithm \cite{mcmahan03do}. 
These policies are frozen for the duration of the experiments.
Further details of the environment and experiments are in Section~\ref{sec:details:soccer}.

\subsubsection{Empirical Verification of Q-Mixing}
\label{sec:soccer:verify}
We now turn to our first question: whether QMVI can obtain Q-values for mixed-strategy opponents. 
To answer this, we run the QMVI algorithm against a fixed opponent mixed strategy (Algorithm~\ref{alg:vi}).
We construct dynamics models for each opponent by considering the opponent's policy and the multiagent dynamics model as a single entity.
Then we may approximate the relative observation-occupancy distribution by rolling out 30 episodes against each policy and estimating the distribution.

In our experiment, an optimal policy was reached in fourteen iterations.
The resulting policy best-responded to each individual opponent and the mixture.
This empirically validates our first hypothesis.

\subsubsection{Coverage of Opponent Strategy Space}
Our second question is whether Q-Mixing can produce high-quality responses for any opponent mixture.
Our evaluation of this question employs the same five opponent pure strategies as the previous experiment.
We first trained a baseline, \emph{BR(Uniform)} \NEW{or BR$(\sigma_{-i})$}, directly against the uniform mixed-strategy opponent.
The baseline was trained using 300000 simulation steps.
The same hyperparameters were used to train against each of the opponent's pure strategies, with the simulation budget split equally. 
The Q-values trained respectively are used as the components for Q-Mixing, and an OPC is also trained from their replay buffers.
The OPC has the same neural network architecture as the policies, but modifies the last layer to predict over the size of the opponent's strategy set.

We evaluate each method against all opponent mixtures truncated to the tenths place {(e.g., [0.1, 0.6, 0.1, 0.2, 0.0])}, resulting in 860 strategy distributions.
This is meant to serve as a representative coverage of the entire opponent's mixed-strategy space. 
For each one of these mixed-strategies, we simulate the performance of each method against that mixture for thirty episodes. 
We then collect the average payoff against each opponent mixture and sort the averages in descending order.

Figure~\ref{fig:coverage} shows Q-Mixing's performance across the opponent mixed-strategy space.
Learning in this domain is fairly easy, so both methods are expected to win against almost every mixed-strategy opponent.
Nevertheless, Q-Mixing generalizes across strategies better, albeit with slightly higher variance.
While the improvement of Q-Mixing is incremental, we interpret this first evaluation as validating the promise of Q-Mixing for coverage across mixtures.

\begin{figure}[ht]
    \centering
    \includegraphics[width=\columnwidth]{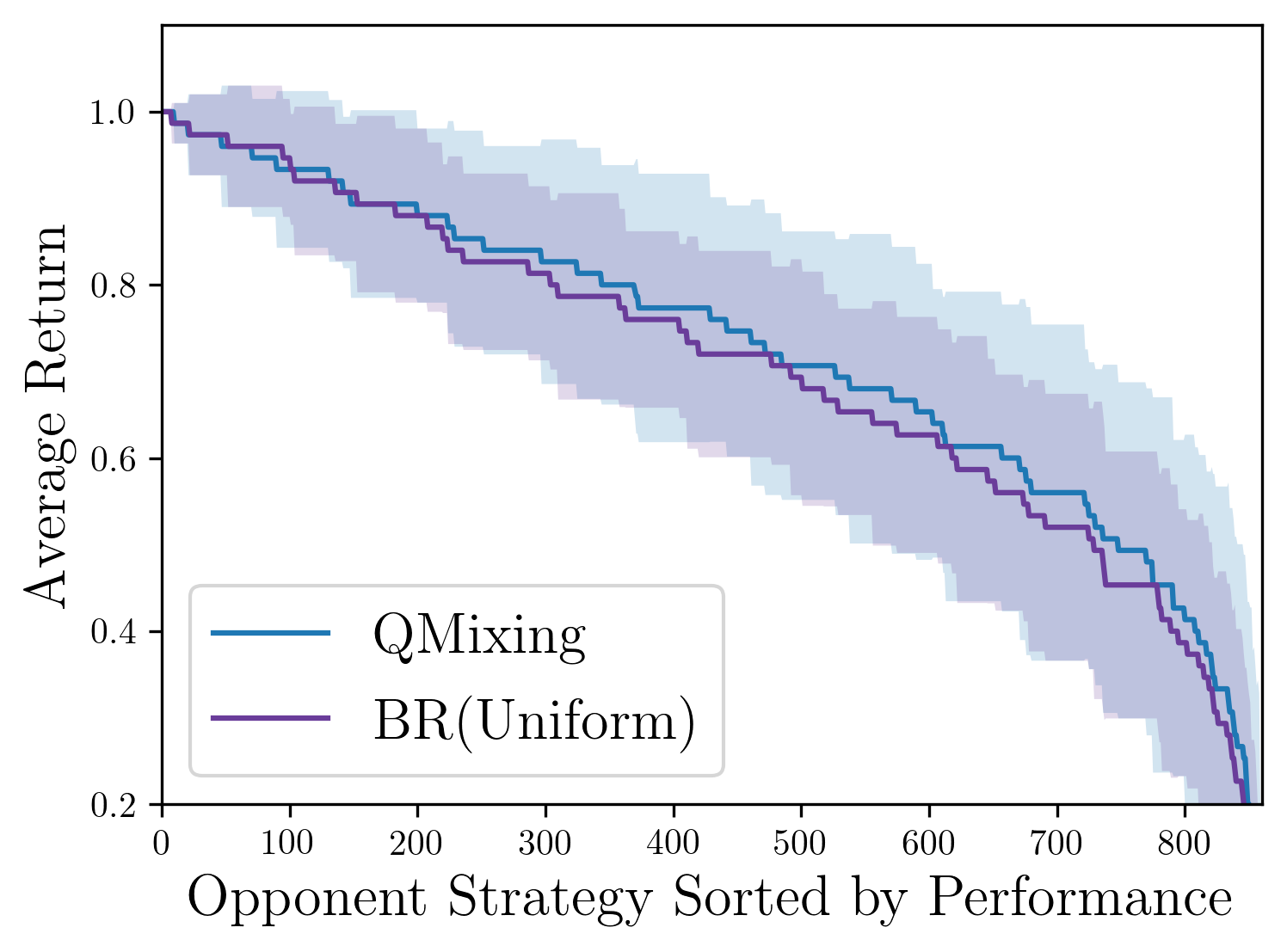}
    \caption{
        Q-Mixing's coverage of the opponent's strategy space \NEW{in the soccer game.} 
        The strategies are sorted \NEW{per-BR-method by the BR's return}.
        Shaded region represents a 95\% confidence interval over five random seeds ($\mathit{df}=4,~t=2.776$).
        BR(Mixture) is trained against the uniform mixture opponent.
        The two methods use the same number of experiences.
    }
    \label{fig:coverage}
\end{figure}

\subsection{Sequential Social Dilemma Game}
\NEW{%
In this section, we evaluate the performance of Q-Mixing in a more complex setting, in the domain of sequential social dilemmas (SSD). 
These dilemmas present each agent with difficult strategic decision where they must balance collective and self interest throughout repeated interactions. 
Like our grid-world soccer example, our SSD \emph{Gathering} game has two players, but unlike the simpler game it is general sum.
Most importantly, the game exhibits imperfect observation, that is, the environment is only partially observable and the players have private information.}
We investigated the following research questions on this environment:
(1)~Can Q-Mixing work with Q-values based on observations from a complex environment? 
(2)~Can the use of an opponent policy classifier that updates the opponent distribution improve performance?
(3)~Can policy distillation be used to compress a Q-Mixing based policy, while preserving performance?

\NEW{%
The Gathering game is a gridworld instantiation of a tragedy-of-the-commons game~\cite{perolat17gathering}. 
In this game, the players compete to harvest apples from an orchard; however, the regrowth rate of the apples is proportional to the number of nearby apples.
Limited resources puts the players at odds, because they must work together to maintain the orchard's health while selfishly amassing apples.
In addition to picking apples, the players may tag all players in a beam in front of them removing them from the game for a short fixed duration.
Moreover, the agents face partial observation as they are only able to see a small window in front of them and do not know the full state of the game.
}

\NEW{%
A visualization of the environment and more details are included in Section~\ref{sec:details:gathering}.
The policies used in this environment are implemented as deep neural networks with two hidden layers with 50 units and ReLU activations.
rAgents can choose to move in the four cardinal directions, turn left or right, tag the other player, or perform no action. 
The policies are trained using the Double DQN algorithm~\cite{hasselt16ddqn}, and the opponent's parameters are always kept fixed (no training updates are performed) for the duration of the experiments.}

All statistical results below are reported with a 95\% confidence interval based on the Student's $t$-distribution ($\mathit{df}=4,~t=2.776$). 
To generate this interval, we perform almost the entire experimental process, as described in each experiment's subsection, entirely across five random seeds.
\NEW{For each seed we do not resample hyperparameters nor generate new opponents.}
A final result is produced for each random seed, often resulting in a sample of payoffs for a variety of player profiles.
Then the sample statistics from each random seed are utilized to construct the confidence interval.

\subsubsection{Empirical Verification of Q-Mixing}
We experimentally evaluate Q-Mixing\NEW{-Prior} on the Gathering game, allowing us to confirm our algorithm's robustness to complex environments. 
First, a set of three opponent policies are generated through \NEW{DO.}
A BR is trained against each of the independent pure-strategies. 
A baseline BR$(\sigma_{-i})$ is trained directly against \NEW{the uniform mixture of the same opponents}.
\NEW{%
We evaluate Q-Mixing-Prior's ability to transfer strategic knowledge by first simulating its performance against the mixed-strategy opponent.
Then we compare this performance to the performance that BR$(\sigma_{-i})$'s strategy that was learned by training directly against the mixed-strategy.
}

\NEW{%
The training curves for $\text{BR}(\sigma_{-i})$ is presented in Figure~\ref{fig:qmixing_training_curve}. 
From this graph we can see that Q-Mixing-Prior is able to achieve performance stronger than $\text{BR}(\sigma_{-i})$ through transfer learning.
This is possible because the pure-strategy BR are able to learn stronger policies from specialization compared to the mixed-strategy BR. 
For example, when playing only against $\pi_{-i}^0$ the best-response $\text{BR}(\pi_{-i}^0)$  achieves a return of $227.93\pm21.01$ while $\text{BR}(\sigma_{-i})$ achieves $207.98\pm15.92$.
This verifies our hypothesis that Q-Mixing is able to transfer knowledge under partial observation. }

\begin{figure}[!ht]
    \centering
    \includegraphics[width=\columnwidth]{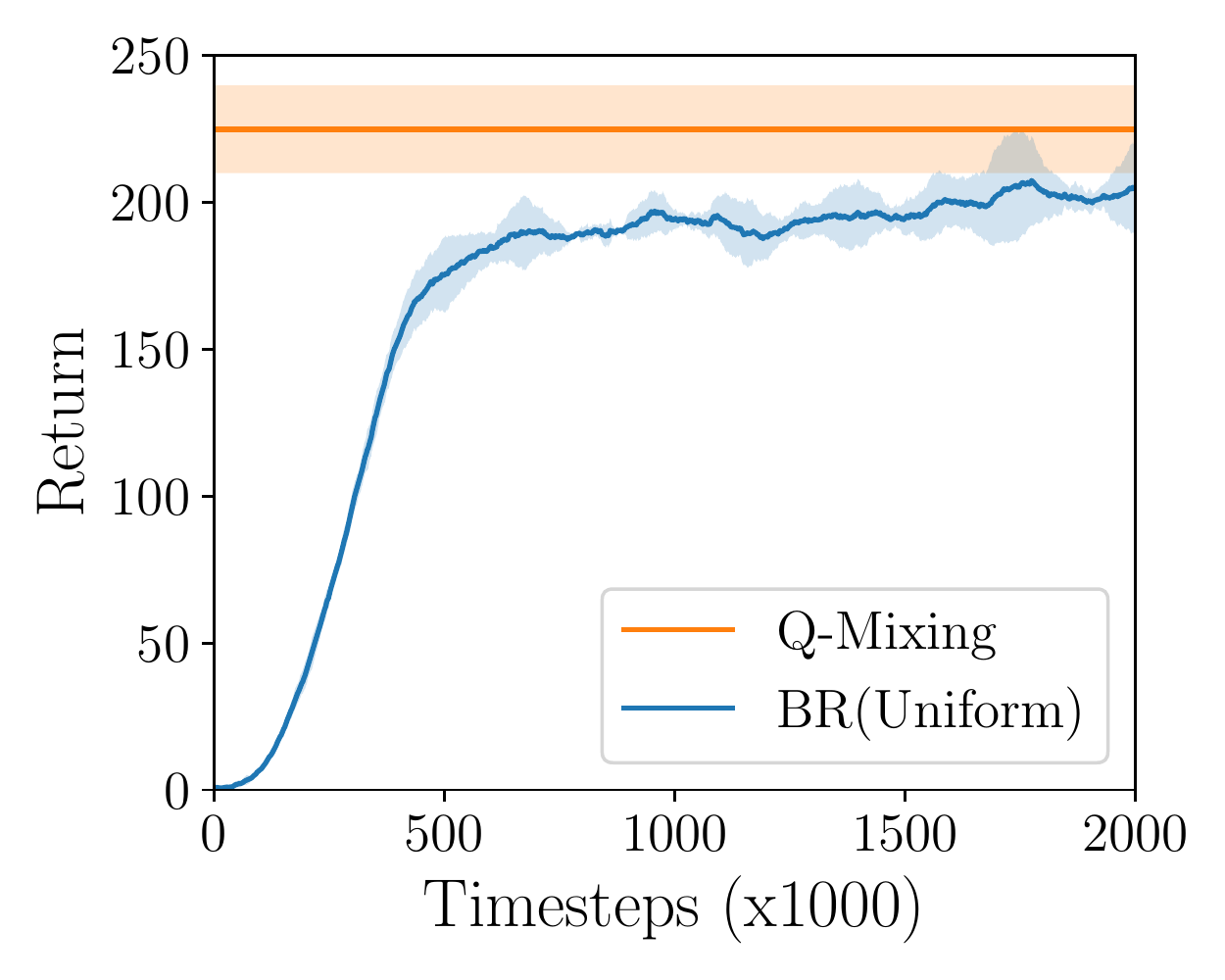}
    \caption{
    \NEW{%
        Learning curve of BR($\sigma_{-i}$) compared to the performance of Q-Mixing-Prior on the Gathering game.
        Shaded region represents a 95\% confidence interval over five random seeds ($\mathit{df}=4,~t=2.776$).}
    }
    \label{fig:qmixing_training_curve}
\end{figure}

\subsubsection{Opponent Classification}
\NEW{%
Our next research question is: can the use of an OPC that updates the opponent distribution in Q-Mixing improve its performance? 
During play against an opponent sampled from the mixed-strategy, the player is able to gather evidence about which opponent they are playing. 
We hypothesize that leveraging this evidence to weight the importance of the respective BR's Q-values higher will improve Q-Mixing's performance.
}

To verify this hypothesis, we train an OPC using the replay buffers associated with each BR policy.
These are the same buffers that were used to train the BRs, and cost no additional compute to collect.
This data is used to train an OPC that outputs a distribution over opponent pure strategies for each observation.
The OPC is implemented with a deep neural network two hidden layers with \NEW{50 units and ReLU activations.} 
To train this classifier we take each experience from the pure-strategy BR replay buffers and assign them each a class label for each respective pure-strategy.
The classifier is then trained to predict this label using a cross-entropy loss. 

\NEW{%
We evaluate Q-Mixing-OPC by testing the performance on a representative coverage of the mixed-strategy opponents illustrated in Figure~\ref{fig:gathering_coverage}.
We found that the Q-Mixing-OPC policy performed stronger against the full opponent strategy coverage supporting our hypothesis that an OPC can identify the opponent's pure-strategy and enable Q-Mixing to chose the correct BR policy.
However, our method has a much larger variance, which can be explained by the OPC's misclassification of the opponent resulting in poorly chosen actions throughout the episode.
}

\begin{figure}[!ht]
    \centering
    \includegraphics[width=\columnwidth]{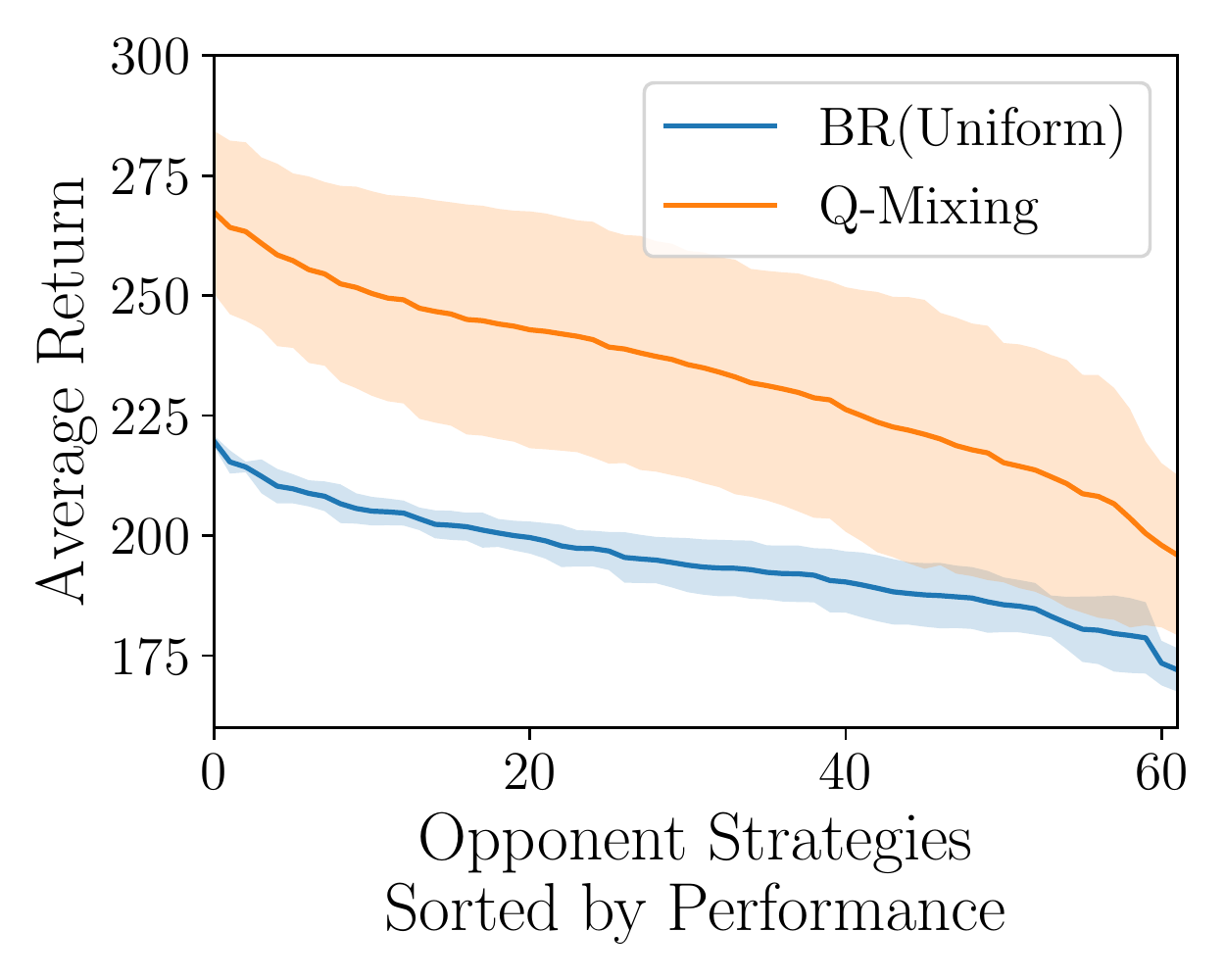}
    \caption{
    \NEW{%
        Q-Mixing with an OPC's coverage of the opponent's strategy space on the Gathering game. 
        Each strategy is a mixture over the 3 opponent's truncated to the tenths place.
        The strategies are sorted by the respective BR's performance.
        Shaded regions represent a 95\% confidence interval over five random seeds ($\mathit{df}=4,~t=2.776$).}
    }
    \label{fig:gathering_coverage}
\end{figure}

\subsubsection{Policy Distillation}
In Q-Mixing we need to compute Q-values for each of the opponent's pure strategies.
This can be a limiting factor in parametric policies, like deep neural networks, where our policy's complexity grows linearly in the size of the support of the opponent's mixture.
This can become unwieldy in both memory and computation. 
To remedy these issues, we propose using policy distillation to compress a Q-Mixing policy into a smaller parametric space \cite{hinton14distill}. 

In the policy distillation framework, a larger neural network referred to as the \emph{teacher} is used as a training target for a smaller neural network called the \emph{student}.
In our experiment, the Q-Mixing policy is the teacher to a student neural network that is the size of a single best-response policy.
The student is trained in a supervised learning framework, where the dataset is the concatenated replay buffers from training pure-strategy best-responses.
This is the same dataset that was used in opponent classifying, which was notably generated without running any additional simulations.
A batch of data is sampled from the replay-buffer and the student predicts $Q^S$ the teacher's $Q^T$ Q-values for each action.
The student is then trained to minimize the KL-divergence between the predicted Q-values and the teacher's true Q-values.
There is a small wrinkle, the policies produce Q-values, and KL-divergence is a metric over probability distributions.
To make this loss function compatible, the Q-values are transformed into a probability distribution by softmax with temperature~$\tau$.
The temperature parameter allows us to control the softness of the maximum operator.
A high temperature produces actions that have a near-uniform probability, and as the temperature is lowered the distribution concentrates weight on the highest Q-Values~\cite{sutton18rl}.
The benefit of a higher temperature is that more information can be passed from the teacher to the student about each state. 
Additional training details are described in Section~\ref{sec:details:distill}.

The learning curve of the student is reported in Figure~\ref{fig:distill}. 
\NEW{%
We found that the student policy was able to recover the performance of Q-Mixing-Prior, albeit with slightly higher variance. 
This study did not include any effort to optimize the student's performance, thus further improvements with the same methodology may be possible.
}
This result confirms our hypothesis that policy distillation is able to effectively compress a policy derived by Q-Mixing.

\begin{figure}[!ht]
    \centering
    \includegraphics[width=\columnwidth]{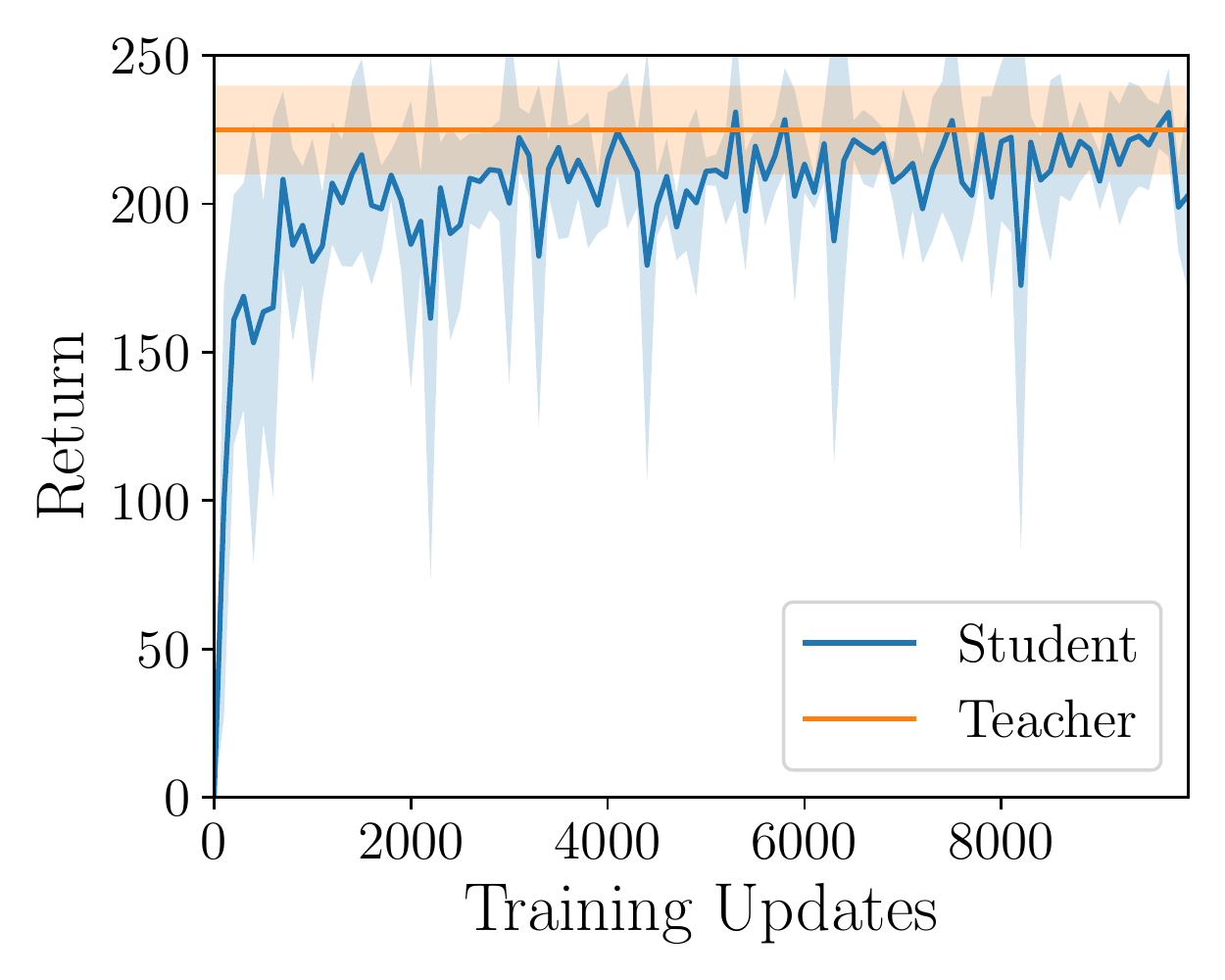}
    \caption{
        \NEW{%
        Policy distillation simulation performance over training on the Gathering game.
        The teacher Q-Mixing-Prior is used as a learning target for a smaller student network.
        Shaded region represents a 95\% confidence interval over five random seeds ($\mathit{df}=4,~t=2.776$).
        }
    }
    \label{fig:distill}
\end{figure}

\section{Related Work}

\subsection{Multiagent Learning}
The most relevant line of work in multiagent learning studies the interplay between centralized learning and decentralized execution \cite{tan93, kraemer14}. 
In this regime, agents are able to train with additional information about the other agents that would not be available during evaluation (e.g., the other agent's state~\cite{rashid18qmix}, actions~\cite{claus98}, or a coordination signal~\cite{greenwald03correlatedq}). 
The key question then becomes: how to create a policy that can be evaluated without additional information? 
A popular approach is to decompose the joint-action value into independent Q-values for each agent \cite{guestrin01, he16opponentmodeling, sunehag17vdn, rashid18qmix, mahajan19maven}.
An alternative approach is to learn a centralized critic, which can train independent agent policies \cite{gupta17, lowe17maddpg, foerster18coma}.
Some work has proposed constructing meta-data about the agent's current policies as a way to reduce the learning instability present in environments where other agents' policies are changing \cite{foerster17replay, omidshafiei17}. 

A set of assumptions that can be made is that all players have fixed strategy sets.
Under these assumptions, agents could maintain more sophisticated beliefs about their opponent~\cite{zheng18}, and extend this to recursive-reasoning procedures~\cite{yang19}.
These lines of work focus more on other-player policy identification and are a promising future direction for improving the quality of the OPC\@.
One more potential extension of the OPC is to consider alternative objectives.
Instead of focusing exclusively on predicting the opponent, in safety-critical situations an agent will want to consider an objective that accounts for inaccurate prediction of their opponent.
The Restricted Nash Response~\citet{johanson07} encapsulates this measure by balancing maximal performance if the prediction is correct balanced with reasonable performance if the prediction is inaccurate.
While both of these directions of work focus around opponent-policy prediction, they do so under a largely different problem statement.
Most notably, these works do not consider varying the distribution of the opponent policies, and as such extending these works to fit this new problem statement would constitute its own study.

Instead of building or using complete models of opponents, one may use an implicit representation of their opponents.
By choosing to build an explicit model of their opponent they circumvent needing a large amount of data to reconstruct the opponent policy.
An additional benefit is that there are less likely to be errors in the model that need to be overcome, because a perfect reconstruction of a complex policy is no longer necessary.
\citet{he16opponentmodeling} proposes DRON which uses a learned latent action prediction of the opponent as conditioning information to the policy (in a similar nature to the opponent-actions in the joint-action value area).
They also show another version DRON which uses a Mixture-of-Experts \cite{jacobs91experts} operation to marginalize over the possible opponent behaviors. 
More formally, they compute the marginal $\sum_{a_{-i}}\pi_{-i}(a_{-i}|s_{i})Q_{i}(s_{i}, a_{i}, a_{-i})$, which is over the action-space and utilized to condition the expected Q-value.
Q-Mixing is built off a similar style of marginalization;  however, it marginalizes over the policy-space of the opponent instead of the action-space.
Moreover, Q-Mixing depends on independent BR Q-values against each opponent-policy, where DRON learns a single Q-network.
\citet{bard13} proposes implicitly modelling opponents through the payoffs received from playing against a portfolio of the agent's policies. 

Most multiagent learning work focuses on the simultaneous learning of many agents, where there is not a distribution over a static set of opponent policies.
This difference in methods can have strong influences on the final learned policies.
For example, when a policy is trained concurrently with another particular opponent-policy they may overfit to each other's behavior.
As a result, the final learned policy may be unable to coordinate or play well against any other opponent policy~\cite{bard19hanabi}
Another potential problem is that each agent now faces a dynamic learning problem, where they must learn a moving target (the other agent's policy) \cite{foerster18lola, tesauro03hyperq}.

\subsection{Multi-Task Learning}
Multiagent learning is analogous to multi-task learning.
In this reconstruction, each strategy/policy is analogous to solving a different task. 
And the opponent's strategy would be the distribution over tasks. 
Similar analogies to tasks can be made with objectives, goals, contexts, etc.~\cite{kaelbling93goals, ruder17multitask}.

The multi-task community has roughly separated learnable knowledge into two categories~\cite{snel14}. 
\emph{Task relevant} knowledge pertains to a particular task \cite{jong05, walsh06}; meanwhile, \emph{domain relevant} knowledge is common across all tasks \cite{caruana97, foster02domain, konidaris06}. 
Work has been done that bridges the gap between these settings; for example, knowledge about a task could be a curriculum to utilize over tasks~\cite{czarnecki18mixmatch}.
In task relevant learning, a leading method is to identify state information that is irrelevant to decision making, and abstract it away \cite{jong05, walsh06}.
Our work falls into the same task relevant category, where we are interested in learning responses to particular opponent policies. 
What differentiates our work from the previous work is that we learn Q-values for each task independently, and do not ignore any information. 

Progressively growing neural networks is another similar line of work~\cite{rusu16progressive}, focused on a stream of new tasks.
\citet{schwarz18progcomp} also found that network growth could be handled with policy distillation.

\subsection{Transfer Learning}
Transfer learning is the study of reusing knowledge to learn new tasks/domains/policies. 
Within transfer learning, we look at either how knowledge is transferred, or what kind of knowledge is transferred. 
Previous work on how to transfer knowledge has tended to follow one of two main directions \cite{pan10survey, lampinen19taskgeneral}.
The \emph{representation transfer} direction considers how to abstract away general characteristics about the task that are likely to apply to later problems. 
\citet{ammar15transfer} present an algorithm where an agent collects a shared general set of knowledge that can be used for each particular task.
The second direction directly transfers parameters across tasks; appropriately called \emph{parameter transfer}. 
\citet{taylor05} show how policies can be reused by creating a projection across different tasks' state and action spaces.

In the literature, transferring knowledge about the opponent's strategy is considered intra-agent transfer \cite{silva19transfer}.
The focus of this area is on \emph{adapting to other agents}. 
One line of work in this area focuses on ad hoc teamwork, where an agent must learn to quickly interact with new teammates \cite{barrett15, bard19hanabi}. 
The main approach relies on already having a set of policies available, and learning to select which policy will work best with the new team \cite{barrett15}. 
Another work proposes learning features that are independent of the game, which can either be qualities general to all games or strategies \cite{banerjee07}. 
Our study differs from these in its focus on the opponent's policies as the source of information to transfer.

\section{Conclusions}

This paper introduces Q-Mixing, an algorithm for transferring knowledge across distributions of opponents. 
We show how Q-Mixing relies on the theoretical relationship between an agent's action-values, and the strategy employed by the other agents. 
A first empirical confirmation of the approach is demonstrated using a simple grid-world soccer environment. 
In this environment, we show how experience against pure strategies can transfer to construction of policies against mixed-strategy opponents. 
Moreover, we show that this transfer is able to cover the space of mixed strategies with no additional computation. 

Next, we tested our algorithm's robustness on a sequential social dilemma.
In this environment, we show the benefit of introducing an opponent policy classifier, which uses the agent's observations to update its belief about the opponent's policy.
This updated belief is then used to revise the weighting of the respective BR Q-values.

Finally, we address the concern that a Q-Mixing policy may become too large or computationally expensive to use.
To ease this concern we demonstrate that policy distillation can be used to compress a Q-Mixing policy into a much smaller parameter space.

\bibliography{bibliography}
\bibliographystyle{icml2021}

\clearpage
\appendix
\section{Q-Mixing}
\label{sec:q-mixing-proof}

\setcounter{theorem}{0}
\begin{theorem}[Single-State Q-Mixing]
\label{thm:qmix}
    Let $Q^*_i(\cdot \mid \pi_{-i})$, $\pi_{-i}\in\Pi_{-i}$, denote the optimal strategic response Q-value against opponent policy $\pi_{-i}$. 
    Then for any opponent mixture $\sigma_{-i} \in\Delta(\Pi_{-i})$, the optimal strategic response Q-value is given by
\begin{displaymath}
Q^*_i(a_i\mid\sigma_{-i}) = \sum_{\pi_{-i}\in\Pi_{-i}} \sigma_{-i}(\pi_{-i}) Q^*_i(a_i\mid\pi_{-i}).
\end{displaymath}
\end{theorem}

\begin{proof}    

The definition of Q-value is as follows \cite{sutton18rl}:
\begin{displaymath}
Q_i^*(a_i)=\sum_{r_i}p(r_i \mid a_i)\cdot r_i.
\end{displaymath}

In a multiagent system, the dynamics model $p$ suppresses the complexity introduced by the other agents. 
We can unpack the dynamics model to account for the other agents as follows:
\begin{displaymath}
p(r_i \mid a_i)=\sum_{\pi_{-i}}\sum_{a_{-i}}\pi_{-i}(a_{-i})\cdot p(r_i \mid \bm{a}).
\end{displaymath}

We can then unpack the strategic response value as follows:
\begin{displaymath}
Q_i^*(a_i \mid \pi_{-i}) = \sum_{a_{-i}} \pi_{-i}(a_{-i}) \sum_{r_i} p(r_i \mid \bm{a}) \cdot r_i.
\end{displaymath}

Now we can rearrange the expanded Q-value to explicitly account for the opponent's strategy. The independence assumption enables the following re-writing by letting us treat the opponent's mixed strategy as a constant condition. 
\begin{align*}
    Q^*_i(a_i &\mid \sigma_{-i}) \\
    &= \sum_{r_i}\sum_{\pi_{-i}}\sigma_{-i}(\pi_{-i}) \sum_{a_{-i}}\pi_{-i}(a_{-i}) p(r_i \mid \bm{a}) \cdot r_i \\
    &=\sum_{\pi_{-i}} \sigma_{-i}(\pi_{-i})\sum_{a_{-i}} \pi_{-i}(a_{-i}) \sum_{r_i} p(r_i \mid \bm{a}) \cdot r_i \\
    &=\sum_{\pi_{-i}} \sigma_{-i}(\pi_{-i}) Q_i^*(a_i \mid \pi_{-i}).
\end{align*}
\end{proof}

\begin{algorithm2e*}[!ht]
\DontPrintSemicolon
\caption{Value Iteration: Q-Mixing}
\label{alg:vi}
\KwIn{$\mathcal{S}, \mathcal{A}, \mathcal{T}, \mathcal{R}, \epsilon, \gamma$} 
$V_0(s\mid\sigma_{-i})\gets \sum_{\pi_{-i}}\sigma_{-i}(\pi_{-i})Q(s,a\mid\pi_{-i})$\; 
\Do{$\exists_{s\in\mathcal{S}}~|V_t(s) - V_{t-1}(s)| > \epsilon$}{
    $Q_t(s, a\mid\sigma_{-i})\gets \sum_{\pi_{-i}} \psi(\pi_{-i} \mid s, \sigma_{-i}) \sum_{s',~r} \mathcal{T}(s', r \mid s, a, \pi_{-i}) \left[ r + \gamma V_{t-1}(s'\mid\sigma_{-i}) \right]$\;
    $V_t(s\mid\sigma_{-i})\gets \max_a Q_t(s,a\mid\sigma_{-i})$\;
    $\pi_t(s\mid\sigma_{-i})\gets \argmax_a Q_t(s,a\mid\sigma_{-i})$\;
}
\KwOut{$V_t,~Q_t,~\pi_t$}
\end{algorithm2e*}

\section{Experimental Details}
\label{sec:details}
Double DQN was used for all experiments~\cite{hasselt16ddqn}. 
Determining the correct hyperparameters to utilize is a non-trivial problem, because the learning dynamics may vary given different opponent policies. 
To this end, we sought to find a method for selecting hyperparameters that performs well against a diversity in opponents, while also being computationally tractable to run.
We choose hyperparameters that performed best against a uniform-mixture of a fixed set of opponents (five for the soccer environment, and three for cyber-security environment). 
The opponent policies were generated through PSRO, and were sampled from the resulting strategy sets.
For both experiments we also included a random opponent to each strategy set, because we expect this opponent to be one of the more challenging opponents to learn against.

However, there's a chicken-and-egg problem present.
In order to utilize PSRO we would also need the aformentioned hyperparameters.
As a stop-gap, we choose initial hyperparameters, by evaluating their performance against a random opponent. 
In summary, for both environments we select hyperparameters by:
\begin{enumerate}
    \item Sample 200 possible hyperparameter settings, and choose the one that best-performs against a random opponent.
    \item Run PSRO until it exceeds a three day walltime.
    \item Sample a fixed set of policies from the strategy set generated from PSRO. 
    \item Sample 200 hyperparameter settings, and evaluate them against uniform mixed-strategy of the four PSRO policies and the random opponent.
\end{enumerate}

We chose to evaluate our hyperparameters against the mixed-strategy opponent, because we believed it offered the most benefit to the baseline method. 
Future work could look at the interplay of the hyperparameter selection method and the respective performance of both Q-Mixing, and learning a BR directly against a mixed-strategy.

The non-standard hyperparameters listed throughout the appendix are defined as follows:
\begin{description}[leftmargin=1cm,labelindent=1cm]
    \item[Timesteps] Total number of experiences collected during training.
    \item[Exploration Fraction] Fraction of the training timesteps used for exploration. The exploration policy is $\epsilon$-greedy, and starts with $\epsilon=1.0$, and linearly decays to \emph{Exploration Final} hyperparameter.
    \item[Exploration Final $\epsilon$] The final $\epsilon$ value.
    \item[Training Frequency] Timestep frequency for performing updates.
    \item[Training Starts] Number of timesteps experienced before training begins.
    \item[Number of Simulations] The number of simulated episodes performed for evaluation.
\end{description}

\subsection{Soccer}
\label{sec:details:soccer}
The soccer environment is a gridworld comprised of two players, one ball, and two goals.
The soccer field is a $5\times 4$ matrix where the goals are off-field on the left and right sides of the field.
The ball can spawn in one two positions in the middle of the field, and the players spawn on either side of the spawn points.
A graphical representation of the soccer environment can be see in Figure~\ref{fig:soccer}.
The players are rewarded for moving the ball into the opponent's goal (the one they spawned furthest from). 

\begin{figure}[ht]
    \centering
    \includegraphics[width=0.45\textwidth]{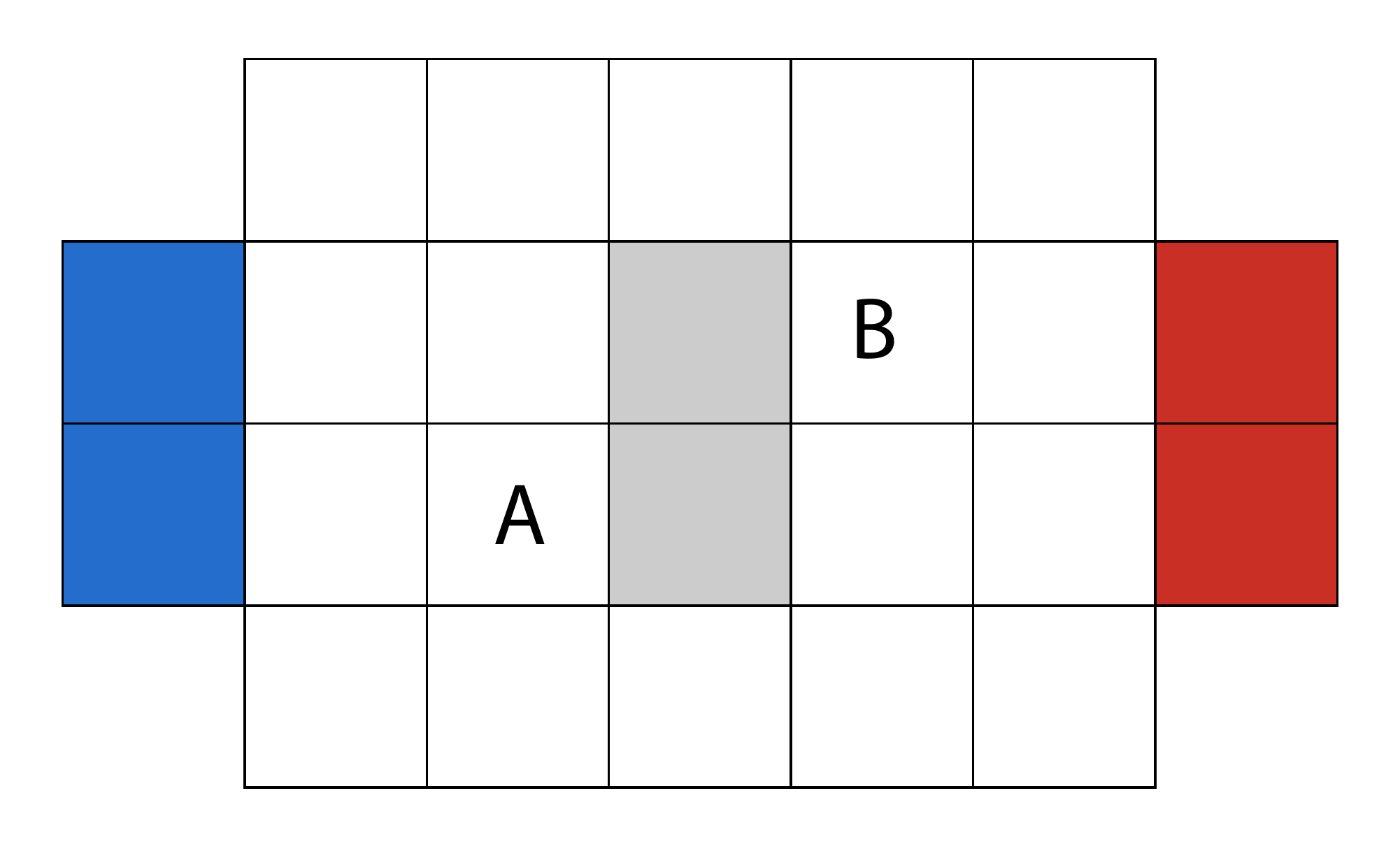}
    \caption{
        Grid-world soccer environment. 
        The letters represent the respective players. 
        The ball may spawn in either middle highlighted tile, and the player's goal is to score in the opposite net.
    }
    \label{fig:soccer}
\end{figure}

In this experiment we select five opponent policies, and hold them fixed throughout the experiments. 
We consider the hyperparameters in Table~\ref{tab:hparam_ranges:soccer} as candidates, and found the hyperparameters listed in Table~\ref{tab:hparams:soccer} to perform best.

To ensure that each method receives the same simulation budget, we allow each pure-strategy BR 60000 timesteps. 
An interesting future direction is investigating the trade-off of simulation budget and performance that exists between these methods.

\begin{table}[ht!]
  \caption{Soccer environment hyperparameters against a mixed strategy.}
  \label{tab:hparams:soccer}
  \centering
  \begin{tabular}{lr}
    \toprule
    Hyperparameter & Value\\
    \midrule
    Optimizer &Adam\\
    Learning Rate &0.0003\\
    Buffer Size &3000\\
    Gamma &0.99\\
    Timesteps &300000\\
    Batch Size &64\\
    Exploration Fraction &0.33\\
    Exploration Final $\epsilon$ &0.01\\
    Training Frequency &1\\
    Training Starts &300\\
    \bottomrule
  \end{tabular}
\end{table}

\begin{table}[ht!]
  \caption{Soccer environment considered hyperparameters.}
  \label{tab:hparam_ranges:soccer}
  \centering
  \begin{tabular}{lr}
    \toprule
    Hyperparameter & Value\\
    \midrule
    Batch Size &32, 64\\
    Buffer Size &300, 1000, 3000, 10000\\
    Learning Rate &1e-3, 3e-3, 1e-4, 3e-4\\
    Timesteps &10000, 30000, 100000, 300000\\ 
    Exploration Fraction &0.1, 0.3, 0.4, 0.7\\
    Training Starts &100, 300, 1000\\
    \bottomrule
  \end{tabular}
\end{table}

The trained DQN was approximated using a 2 hidden layer neural network. 
The hidden layers each had 50 units and were fully connected with a ReLU activation.
The possible actions are moving in any of the four cardinal directions, or staying in place.
The input is a vector of length 120, to represent each of the 20 positions on the board having one of the following states: 
\begin{itemize}
    \item Player 0 is on this square and does not have the ball.
    \item Player 0 is on this square and is holding the ball.
    \item Player 1 is on this square and does not have the ball.
    \item Player 1 is on this square and is holding the ball.
    \item The ball is on the ground on this square.
    \item Unoccupied space.
\end{itemize}

\subsubsection{Opponent Policy Classifier}
The hyperparameters selected for training the opponent classifier are listed in Table~\ref{tab:hparams:soccer_opc}. 
The replay buffers gathered from training best-responses against each opponent were merged into one dataset.
The classifier was trained to predict the opponent for each observation in the dataset.
This resulted in 15000 data points, which were randomly split 90-10 between training and validation.

\begin{table}[ht]
  \caption{Markov-Soccer opponent policy classifier hyperparameters.}
  \label{tab:hparams:soccer_opc}
  \centering
  \begin{tabular}{lr}
    \toprule
    Hyperparameter & Value\\
    \midrule
    Optimizer &Adam \\
    Learning Rate &$5\cdot 10^{-5}$\\
    Loss &Cross Entropy \\
    Batch Size &64 \\
    \bottomrule
  \end{tabular}
\end{table}

The classifier was a neural network with the same architecture as a single policy; however, the last layer is modified to choose opponents rather than actions.
We did not perform a hyperparameter search on this network or learning algorithm.

\subsection{Gathering}
\label{sec:details:gathering}
The Gathering environment is a gridworld tragedy-of-the-commons game. 
For motivation and background on the environment please refer to \citet{perolat17gathering, leibo17ssd}.
The observation space in the Gathering environment is the rectangular area in front of the agent stretching 20 cells forward with a width of 10.
The agents simultaneous take actions of either moving in the four cardinal directions, rotating left or right, tagging the other agent with a time-out beam, or taking no action. 
A visual depiction of the environment is provided in Figure~\ref{fig:gathering_map}.

\begin{figure*}
    \centering
    \includegraphics[width=0.5\textwidth]{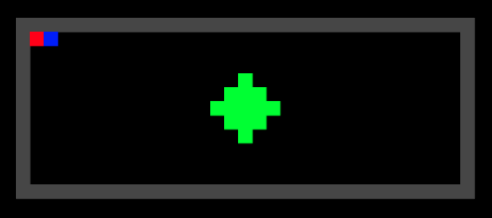}
    \caption{%
        Gathering game map.
        Players spawn in either the blue or red tile, and apples on the green tiles. 
    }
    \label{fig:gathering_map}
\end{figure*}

The hyperparameters that were searched over are provided in Table~\ref{tab:hparams_ranges:gathering} and the final hyperparameters are in Table~\ref{tab:hparams:gathering}.

\begin{table}[ht]
  \caption{Gathering environment hyperparameters.}
  \label{tab:hparams:gathering}
  \centering
  \begin{tabular}{lr}
    \toprule
    Hyperparameter & Value\\
    \midrule
    Optimizer &Adam\\
    Learning Rate &$3\cdot 10^{-4}$\\
    Gradient Norm Clip &None\\
    Buffer Size &30000\\
    Gamma &0.99\\
    Batch Size &64\\
    Exploration Fraction &0.3 \\
    Exploration Final $\epsilon$ &0.03\\
    Training Starts &1000\\
    \bottomrule
  \end{tabular}
\end{table}

\begin{table*}[ht]
  \caption{Gathering environment considered hyperparameters.}
  \label{tab:hparams_ranges:gathering}
  \centering
  \begin{tabular}{lr}
    \toprule
    Hyperparameter & Value\\
    \midrule
    Learning Rate &3e-3, 1e-4, 3e-4, 1e-5, 3e-5\\
    Gradient Norm Clip &None, 0.1, 1, 10\\
    Buffer Size &3e4, 5e4, 8e5, 1e5\\
    Batch Size &32, 64\\
    Timesteps &3e5, 5e5, 7e5, 1e6, 1.2e6, 1.5e6, 1.7e6, 2e6\\
    Exploration Timesteps &1e5, 2e5, 3e5, 4e5, 5e5, 7e5\\
    \bottomrule
  \end{tabular}
\end{table*}

\subsection{Policy Distillation}
\label{sec:details:distill}
In the policy distillation framework, a larger neural network referred to as the ``teacher'' is used as a training signal for a smaller neural network called the ``student''. 
In our experiment the Q-Mixing policy is the teacher to a student neural network that is the size of a single BR policy.
The student is trained via supervised learning, reusing the pure-strategy BRs' replay buffers as a dataset.
A batch of data is sampled from the replay-buffer and the student predicts $Q^S$ the teacher's response $Q^T$. 
The student is trained to imitate the softmax policy of the teacher.
The full policy distillation loss is
$$
\mathcal{L}_{\text{Distill}}=\sum_{i}^{|D|} \text{softmax}(\frac{Q^T}{\tau})\ln{\frac{\text{softmax}(\frac{Q^T}{\tau})}{\text{softmax}(\frac{Q^S}{\tau})}},
$$
where $D$ is the dataset of concatenated replay buffers.

The hyperparameters used in policy distillation are listed in Table~\ref{tab:hparams:distill}. 
The student policy is the same neural network that's used in computing the best-responses to individual policies; it is described in Section~\ref{sec:details:attackgraph}.
We did not perform a hyperparameter search on this network or learning algorithm.

\begin{table}[ht]
  \caption{Policy distillation hyperparameters.}
  \centering
  \label{tab:hparams:distill}
  \begin{tabular}{lr}
    \toprule
    Hyperparameter & Value\\
    \midrule
    Optimizer &Adam \\
    Learning Rate &0.003\\
    Batch Size &64 \\
    \bottomrule
  \end{tabular}
\end{table}

\end{document}